\documentclass[fontsize=11pt, headings=normal, headinclude=true, paper=a4, pagesize, cleardoublepage=current, BCOR=10mm, fleqn, numbers=noenddot]{scrartcl}
\pdfoutput=1
\usepackage{lmodern}
\setkomafont{sectioning}{\normalcolor\bfseries}%Serifen
\recalctypearea % Nach ändern der Schriftart von KOMA-Script ganz bequem den Satzspiegel neu berechnen lassen.
\setcapindent{0pt}
  \usepackage[english]{babel}  % mehrsprachiger Textsatz
  \usepackage[numbers,square]{natbib}
  \usepackage{mathrsfs}
  \usepackage{amsmath}
  \usepackage{amssymb}
  \usepackage{amsthm}
  \usepackage{enumitem}
  \usepackage[utf8x]{inputenc}    % Passen Sie die Option an.
  \usepackage[T1]{fontenc}         % Schriftenkodierung
  \usepackage{graphicx}            % zum Einbinden von Grafiken
  \usepackage{color}		%Farben, für Kommentare
  \usepackage[english=british]{csquotes} %autostyle
\newtheorem{thm}{Theorem}[section]
\newtheorem{prop}[thm]{Proposition}
\newtheorem{lem}[thm]{Lemma}

\newtheorem*{thm*}{Theorem}
\newtheorem*{Kor*}{Korollar}
\theoremstyle{definition}

\newcommand{\norm}[1]{\ensuremath{\lVert #1 \rVert}}
\newcommand{\abs}[1]{\ensuremath{\lvert #1 \rvert}}
\newcommand{\eps}{\varepsilon}

\renewcommand{\epsilon}{\varepsilon}
\newcommand{\Opeps}[1]{\ensuremath{\mathrm{Op}_\eps^W\left( #1 \right)}}
\newcommand{\ii}{\infty}
\newcommand{\cV}{\mathcal{V}}
\newcommand{\cE}{\mathcal{E}}
\newcommand{\ud}{\mathrm{d}}
\newcommand{\ue}{\mathrm{e}}
\newcommand{\ui}{\mathrm{i}}
\newcommand{\R}{\mathbb{R}}
\newcommand{\C}{\mathbb{C}}
\newcommand{\N}{\mathbb{N}}

\DeclareMathOperator{\tr}{Tr}
\DeclareMathOperator{\Tr}{Tr}

\newcommand\pscal[1]{{\ensuremath{\left\langle #1 \right\rangle}}}
\begin{document}
\title{Semi-classical Dirac vacuum polarisation in a scalar field}
\author{Jonas Lampart and Mathieu Lewin}
\maketitle
\begin{abstract}
We study vacuum polarisation effects of a Dirac field coupled to an external scalar field and derive a semi-classical expansion of the regularised vacuum energy. The leading order of this expansion is given by a classical formula due to Chin, Lee-Wick and Walecka, for which our result provides the first rigorous proof. We then discuss applications to the non-relativistic large-coupling limit of an interacting system, and to the stability of homogeneous systems.
\end{abstract}
\section{Introduction}
The three-dimensional, massive Dirac operator in a scalar field $\phi$ is given by
\begin{equation*}
D_\phi=-\ui \sum_{j=1}^3 \alpha_j \frac{\partial}{\partial x_j}  + \beta (1 + \phi)\,,
\end{equation*}
where $\alpha_j,\beta \in \C^{4\times 4}$ are a representation of the anti-commuting Dirac-matrices with $\alpha_j^2=1$, $\beta^2=1$ and we have chosen units such that $\hbar=c=m=1$. Under appropriate assumptions on $\phi$ it is a self-adjoint operator with domain $H^1(\R^3, \C^4)\subset L^2(\R^3, \C^4)$.
It arises, for instance, in nuclear physics, where the interaction between nucleons is solely modelled by various \emph{meson} fields, including the scalar $\sigma$-meson $\phi$.
According to the picture of the Dirac-sea, a particle, such as a nucleon, should be described by a function $\psi$ in the positive spectral subspace for $D_\phi$ in combination with the filled \enquote{sea} of negative energy states. The state of this particle is then described by its one-particle density matrix
\begin{equation*}
 \gamma:=\chi_{(-\infty, 0]} (D_\phi) + \vert \psi \rangle \langle \psi\vert\,,
\end{equation*}
where $\chi_I$ denotes the characteristic function of $I$. In this picture, the spectral projection $\chi_{(-\infty, 0]} (D_\phi)$ describes a state with no particles, the vacuum. The total energy is then (formally) given by
\begin{equation*}
 \tr(D_\phi\gamma)=-\frac12\tr|D_\phi|+\pscal{\psi,D_\phi,\psi}\,.
\end{equation*}
The first term is the formal expression for the energy of the vacuum since the spectrum of $D_\phi$ is symmetric with respect to zero. 
% 
% 
% 
% Minimizing over $\psi$ with $\phi$ fixed gives the first eigenfunction of $D_\phi$. In general, for $N$ nucleons, the state of lowest energy will always be
% $$\gamma_\mu=\chi_{(-\infty, \mu]} (D_\phi)$$
% where $\mu$ is chosen so that
% $$\tr\chi_{[0, \mu]} (D_\phi)=N.$$
% It is often convenient to think of the chemical potential $\mu$ as being given instead of the particle number $N$.
% The energy of the scalar field $\phi$ is given by
% \begin{equation*}
% \frac12\int_{\R^3}|\nabla\phi(x)|^2\,dx+\frac{M^2}2\int_{\R^3}|\phi(x)|^2\,dx.
% \end{equation*}
% A stationary point of the total energy solves a system of coupled equations
% which, after inserting the formula of $\gamma_\mu$ leads to a non-linear equation for $\phi$.
% 
% 
Of course, $-(1/2)\tr|D_\phi|$ is not finite and must be regularised. 

Similar objects arise if one considers instead of the scalar field $\beta \phi$ an electrostatic field $\phi$. Using various methods of regularisation, these have been studied rigorously by many authors, especially from a variational point of view~\cite{ChaIra-89, ChaIraLio-89, BacBarHelSie-99,HaiSie-03, HaiLewSer-05a, HaiLewSer-05b, HaiLewSol-07, HaiLewSer-09}. The charge density of the vacuum was derived in a perturbative regime by Hainzl~\cite{Hainzl-04}. 

Of the several possible regularisation procedures we will follow~\cite{Chin-77,SerWal-86,Walecka-04} and define the regularised vacuum energy by
\begin{equation}
 \mathcal{E}_{\mathrm{vac}}(\phi)=\Tr\left(-\tfrac12|D_\phi| - R(\phi)\right)\,,
 \label{eq:def_E_vac}
\end{equation}
where $R(\phi)$ is the fourth order Taylor expansion of $-\tfrac12 |D_\phi|$ at $\phi=0$. It will be proved in Proposition~\ref{prop:Evac} below that $-\tfrac12 |D_\phi| - R(\phi)$ is trace-class under the sole assumption that $\phi\in H^1(\R^3,\R)$.
This is a reasonable assumption as it it is equivalent to the finiteness of the energy of the $\sigma$-field 
\begin{equation*}
\frac12\int_{\R^3}\left(|\nabla\phi(x)|^2+ M^2|\phi(x)|^2\right)\,dx\,.
\end{equation*}
If we consider the total energy, the sum of the energies of the $\sigma$-field and the Dirac particles, the stationary points solve a system of coupled equations. We discuss applications of our results, obtained for fixed $\phi$, to this problem in Sections~\ref{sect:app} and~\ref{sec:matter}. 
The dynamics of such coupled systems were studied by Sabin~\cite{Sabin-14} with a scalar field solving a non-linear wave equation, and in~\cite{HaiLewSpa-05, Sabin-14} for an electric field. 

The formula~\eqref{eq:def_E_vac} is very complicated to evaluate, even numerically, for a given $\phi$. In applications it is therefore important to obtain simple approximations of $\cE_{\mathrm{vac}}$ in different regimes.
For the study of nuclear matter, Walecka~\cite{SerWal-86,Walecka-04}, Lee-Wick~\cite[p.~2306]{LeeWic-74} and Chin~\cite[Eq.~(3.53)]{Chin-77} calculated the vacuum energy per unit volume for a constant $\phi\in(-1,+\ii)$ and obtained the simple formula
\begin{equation}\label{eq:Walecka}
 \mathcal{V}(\phi)=-\frac{1}{4 \pi^2} (1+\phi)^4\log(1+\phi) - P(\phi)\,,
\end{equation}
where
\begin{equation*}
P(\phi)=-\frac{1}{4 \pi^2}\left(\phi+\frac{7}{2}\phi^2+\frac{13}{3}\phi^3+\frac{25}{12}\phi^4\right)
\end{equation*}
is the fourth-order Taylor polynomial of the first term.
In this paper, we give the first rigorous derivation of the effective vacuum energy~\eqref{eq:Walecka}, in a semi-classical regime where the field $\phi$ varies slowly. Our main result is the following
\begin{thm}[Derivation of the effective vacuum energy]\label{thm:CV}
Let $\phi\in \mathscr{C}^\infty(\R^3, \R)$ be such that
\begin{itemize}
 \item $|\phi|<1$;
 \item $(1+\abs{x}^2)^{(1+\abs{\alpha})/2}\partial_{x}^\alpha\phi\in L^\ii(\R^3)$ for every multiindex $\alpha\in \N^3$.
\end{itemize}
Then we have
\begin{equation}
\lim_{\epsilon\to0} \epsilon^3\mathcal{E}_{\mathrm{vac}}\big(\phi(\epsilon\cdot)\big)=\int_{\R^3} \mathcal{V}(\phi(x))\,\ud x\,,
\label{eq:semi-classics}
\end{equation}
where $\mathcal{V}$ is given by~\eqref{eq:Walecka}.
\end{thm}
The convergence~\eqref{eq:semi-classics} holds under weaker assumptions on $\phi$, but in Section~\ref{sect:semicl} we actually prove a more precise result. Namely, we show the existence of a function $\cV_\epsilon$ which approximates the true vacuum energy to any order, and which coincides with $\cV$ to first order. That is,
\begin{equation*}
 \epsilon^3\mathcal{E}_{\mathrm{vac}}\big(\phi(\epsilon\cdot)\big)=\int_{\R^3} \mathcal{V}_{\eps}(\phi)(x)\,\ud x+O(\eps^n),\qquad \forall n\geq1
\end{equation*}
where $\mathcal{V}_{\eps}=\cV+O(\epsilon^2)$. The function $\cV_\eps$ is a series in $\eps$, each term depending on a finite number of derivatives of $\phi$. We thereby justify as well the higher-order expansions found in later works~\cite{Blunden-90, Chan-86, Cheyette-85, Perry-86}. Explicit formulas of the type~\eqref{eq:Walecka} can be obtained for all odd dimensions. We will focus on the most interesting case of three-dimensional space, the corresponding proofs for other dimensions can easily be obtained by adapting our arguments line-by-line.

As we said, the regularisation used in~\eqref{eq:def_E_vac} is not the only possible choice. Following~\cite{GraHaiLewSer-13}, one can also consider the Pauli-Villars energy
$$\mathcal{E}_{\mathrm{PV}}(\phi):=-\frac12\tr\left(\sum_{j=0}^4c_j\,|D_{m_j,\phi}|\right)$$
which corresponds to introducing new masses $m_j$ and coefficients $c_j$ such that $c_0=m_0=1$ and
$$\sum_{j=0}^4c_j(m_j)^k=0,\qquad \forall k=0,...,3.$$
Our method can be generalized to this case, leading to the convergence
$$\lim_{\eps\to0}\eps^3\mathcal{E}_{\mathrm{PV}}\big(\phi(\eps\cdot)\big)=-\frac{1}{4\pi^2}\int_{\R^3}\left(\sum_{j=0}^4c_j\;(m_j+\phi)^4\log(1+\phi/m_j)\right).$$
This model still depends on the regularisation masses. In the limit where $m_j\to\ii$ for all $j=1,..,3$, one exactly recovers $\int_{\R^3}\cV(\phi)$.

The advantage of the Pauli-Villars scheme is that it preserves gauge invariance. Since the $\sigma$-model studied in this paper has no such invariance, the simpler regularisation procedure~\eqref{eq:def_E_vac} works just as well. Using the Pauli-Villars method will be important for the study of the Dirac vacuum in a slowly varying magnetic field in the forthcoming work~\cite{GraLewSer-15}. In this case, the effective energy obtained in the limit is the famous Heisenberg-Euler Lagrangian~\cite{HeiEul-36,Weisskopf-36,Schwinger-51a}.

In physical quantities the parameter $\eps$ may correspond to $\hbar/mc$, where $m$ is the mass of the Dirac particles, in which case the limit $\eps\to 0$ would be relevant to both the semi-classical limit $\hbar\to 0$ and the non-relativistic limit $c\to\infty$. For an interacting system, for which bound states exhibiting a certain scaling behaviour may exist, other identifications are possible.
As an application, we discuss the simultaneous non-relativistic and large-coupling limit of the Dirac-sigma model in Section~\ref{sect:app}, where $\eps$ is $\sqrt{\hbar/mc}$. Finally, in Section~\ref{sec:matter} we present some numerical results about the stability of a homogeneous system of Dirac particles interacting with a scalar field, a question related to nuclear matter.
\section{The regularised vacuum energy}\label{sect:reg}
From now on we always work with the Dirac operator $D_\phi$ including a semi-classical parameter $\eps$
\begin{equation*}
 D_\phi:=-\ui\eps \sum_{j=1}^3 \alpha_j  \frac{\partial}{\partial x_j}  + \beta (1 + \phi)\,,
\end{equation*}
which is of course equivalent to taking a slowly varying field $\phi_\eps(x)=\phi(\eps x)$, by a unitary dilation.
% 
% 
% Our choice of regularity for $\phi$ is natural if we think of a coupled particle-field problem with energy
% \begin{equation*}
%  \mathcal{E}(\gamma, \phi)=\Tr(D_\phi \gamma) + \int_{\R^3} \abs{\nabla \phi}^2(x) + M^2\abs{\phi(x)}^2\,\ud x\,.
% \end{equation*}
%
In this section we introduce the regularised vacuum energy and derive its basic properties. We suppose $\phi\in H^1(\R^3,\R)$, which ensures that $\beta \phi$ is infinitesimally $D_0$-bounded and thus that this operator is self-adjoint with domain $\mathrm{Dom}(D_\phi)=H^1(\R^3, \C^4)$. We define the regularised vacuum energy by the formula:
\begin{equation}\label{eq:Edef}
 \mathcal{E}_{\mathrm{vac}}(\phi):=\frac12 \Tr \bigg( - \abs{D_\phi} +
 \sum_{j=0}^4 \frac{1}{j!} \frac{\ud^j}{\ud s^j}\bigg\vert_{s=0} \abs{D_{s\phi}}\bigg)\,.
 %\int_\R \sum_{j=0}^4 \frac{1}{j!}\frac{\ud^j}{\ud s^j}\frac{D_{s\phi}^2}{D_{s\phi}^2+\omega^2} \,\ud \omega\bigg)\,.
\end{equation}
Note that for small $s$, zero is not in the spectrum of $D_{s \phi}$ so we can expect that $\abs{D_{s\phi}}$ is differentiable.
Our proof that~\eqref{eq:Edef} is a correct definition relies on the integral formula
\begin{equation}\label{eq:abs_int}
 \abs{D_\phi}=\frac1\pi \int_\R \frac{D_{\phi}^2}{D_{\phi}^2+\omega^2}\,\ud \omega=\frac{1}{2\pi}\int_{\R} 2 - \frac{\ui \omega}{D_\phi + \ui \omega} + \frac{\ui \omega}{D_\phi - \ui \omega}\,\ud \omega\,,
\end{equation}
where the integral (in the sense of Bochner) exists in the norm topology of $\mathscr{L}(H^2, L^2)$, even if zero is in the spectrum of $D_\phi$, because in that norm the integrand decays like $\omega^{-2}$ for $\omega\to \infty$ and stays bounded at $\omega=0$. 
\begin{prop}\label{prop:Evac}
The regularised vacuum energy defined by equation~\eqref{eq:Edef} is a continuous function of $\phi \in H^1(\R^3, \R)$. It equals
\begin{equation*}
 \mathcal{E}_{\mathrm{vac}}(\phi)= -\frac{1}{2\pi}\Tr\bigg( \int_\R \frac{D_\phi^2}{D_\phi^2 + \omega^2}
-\sum_{j=0}^4 \frac{1}{j!} \frac{\ud^j}{\ud s^j}\bigg\vert_{s=0} \frac{D_{s\phi}^2}{D_{s\phi}^2 + \omega^2}\, \ud \omega\bigg)
\end{equation*}
 and is smooth on the open set $U:=\lbrace \phi\in H^1(\R^3,\R): 0\notin  \sigma(D_\phi)\rbrace$.
 %and extends continuously to $H^1(\R^3, \R)$.
\end{prop}
Before proving Proposition~\ref{prop:Evac} let us prove two preparatory lemmas on the resolvent of $D_\phi$. Throughout, we will denote by $\mathfrak{S}_p(\mathscr{H})$, $1\leq p\leq \infty$, the $p$-th Schatten-ideal in the bounded operators $\mathscr{L}(\mathscr{H})$ on the Hilbert space $\mathscr{H}$. The most important cases are the trace class ($p=1$) and the compact operators ($p=\infty$).
\begin{lem}\label{lem:resolvent1}
 For every $\omega \in \R$ and $3<q\leq \infty$ the map $\phi\mapsto \beta\phi(D_0 + \ui \omega)^{-1}$ is continuous from $H^1(\R^3,\R)$ to $\mathfrak{S}_q(L^2(\R^3, \C^4))$.
\end{lem}
\begin{proof}
Using the Kato-Seiler-Simon (KSS)-inequality~\cite[Theorem 4.1]{Simon-79}
\begin{equation}\label{eq:KSS}
 \norm{f(x)g(-\ui \nabla_x)}_{\mathfrak{S}_q(L^2(\R^3))}\leq (2\pi)^{-3/q} \norm{f}_{L^q(\R^3)} \norm{g}_{L^q(\R^3)}\,,
\end{equation}
and the Sobolev embedding $H^1(\R^3)\hookrightarrow L^6(\R^3)$ we have for $3<q\leq 6$:
\begin{align}
  \norm{\beta\phi(D_0 + \ui \omega)^{-1}}_{\mathfrak{S}_q}
  %&\leq C \norm{\phi}_q\norm{(\eps^2 p^2 + \omega^2 +1)^{-1/2}}_{L^q}\notag\\
  %
  &\leq C (1+\omega^2)^{(3-q)/2q}\eps^{-3/q} \norm{\phi}_{H^1} \norm{(p^2+1)^{-1/2}}_{L^q}\label{eq:KSS cont}
 \end{align}
 with some constant $C>0$. Since $\mathfrak{S}_r\hookrightarrow\mathfrak{S}_q$ continuously for $r>q$ this proves the claim.
\end{proof}
\begin{lem}\label{lem:resolvent2}
 The set $U:=\lbrace \phi\in H^1(\R^3,\R): 0\notin  \sigma(D_\phi)\rbrace$ is an open subset of $H^1(\R^3,\R)$ and for every $\omega \in \R$ the map $\phi \mapsto (D_\phi + \ui \omega)^{-1}$ from $U$ to $\mathscr{L}\big((L^2(\R^3, \C^4)\big)$ has infinitely many continuous derivatives.
\end{lem}
\begin{proof}
If $\phi\in U$ and $\xi \in H^1(\R^3)$ has sufficiently small norm, then $\beta \xi D_\phi^{-1}$ has operator norm less than one. Hence $D_\phi +\beta \xi$ is invertible by a Neumann series and $U$ is open. Then, the resolvent formula
 \begin{equation}\label{eq:res_form}
  (D_\phi +\beta \xi\pm \ui \omega)^{-1}-(D_\phi\pm \ui \omega)^{-1}=-(D_\phi +\beta \xi\pm \ui \omega)^{-1}\beta \xi(D_\phi\pm \ui \omega)^{-1}
 \end{equation}
holds and implies continuity. Letting $\xi=s\zeta$ this allows us to calculate the derivative
\begin{align}
  \langle\mathrm{Diff}((D_\phi\pm \ui \omega)^{-1}),\zeta\rangle
&=\lim_{s\to 0} \tfrac1s \Big((D_\phi +\beta s\zeta\pm \ui \omega)^{-1}-(D_\phi\pm \ui \omega)^{-1}\Big)\notag\\
  &=-(D_\phi \pm \ui \omega)^{-1}\beta \zeta(D_\phi\pm \ui \omega)^{-1}\,,\notag
\end{align}
which is obviously also continuous. Using~\eqref{eq:res_form} again shows that convergence to the derivative holds also in the operator norm, and the fact that the derivative is a product of differentiable functions gives smoothness.
\end{proof}
%
%
%We now turn to the proof of Proposition~\ref{prop:Evac}.
%
\begin{proof}[Proof of Proposition~\ref{prop:Evac}]
Consider $\abs{D_{s \phi}}$ given by equation~\eqref{eq:abs_int}. By Lemma~\ref{lem:resolvent2} the integrand is a smooth function of $s$ in a neighbourhood of $s=0$. One easily checks that its derivatives at $s=0$ are elements of $L^1\left((\R, \ud \omega), \mathscr{L}(H^2, L^2)\right)$, so we can exchange integration and differentiation w.r.t.~$s$, which establishes equality of the formula for $\mathcal{E}_{\mathrm{vac}}$ given in the Proposition with~\eqref{eq:Edef}, in the sense that the expressions inside the trace are equal. In order to show that this expression is an operator of trace-class we first calculate
 \begin{equation*}
 \frac{\ud^j}{\ud s^j}\bigg\vert_{s=0}\left(D_{s\phi} \pm \ui \omega\right)^{-1}
 =(-1)^j j!  \left(D_0 \pm \ui \omega\right)^{-1}
\big(\beta \phi \left(D_0 \pm \ui \omega\right)^{-1}\big)^j\,.
\end{equation*}
Then, using~\eqref{eq:abs_int} and the resolvent formula~\eqref{eq:res_form}, we can write the sum of the integrands as $\tfrac12 (I(\omega, \phi) + I(-\omega, \phi))$ with
\begin{align}
I(\omega, \phi):&=-\ui \omega (D_\phi + \ui \omega)^{-1} + \ui \omega \sum_{j=0}^4
 \frac{1}{j!} \frac{\ud^j}{\ud s^j}\bigg\vert_{s=0}(D_{s\phi} + \ui \omega)^{-1}\notag\\
 &=\ui \omega (D_\phi + \ui \omega)^{-1} \big(\beta \phi \left(D_0 + \ui \omega\right)^{-1}\big)^5\label{eq:E integrand}
\,.
\end{align}
Using the estimate~\eqref{eq:KSS cont} with $q=5$ we find
\begin{align}
 \norm{I(\omega, \phi)}_{\mathfrak{S}_1} &\leq  
\norm{\ui \omega(D_\phi + \ui \omega)^{-1}}_{\mathfrak{S}_\infty}\norm{\abs{(\beta \phi \left(D_0 + \ui \omega\right)^{-1}}^5}_{\mathfrak{S}_1}\notag\\
 &\leq C(1+\omega^2)^{-1} \eps^{-3} \norm{\phi}_{H^1}^5\,.\label{eq:Ibound}
\end{align}
This shows that $I(\omega, \phi)\in L^1(\R, \mathfrak{S}_1)$ for all $\phi \in H^1$, so the expression~\eqref{eq:Edef} is indeed well-defined. 

Continuity in $\phi$ is an immediate consequence of Lemma~\ref{lem:resolvent1} and the Dominated Convergence Theorem. 
Since $\langle \mathrm{Diff}(\beta \phi(D_0 +\ui \omega)^{-1})), \xi\rangle= \beta \xi(D_0 +\ui \omega)^{-1})$ the derivatives of $I(\omega, \phi)$ also satisfy an estimate of the form~\eqref{eq:Ibound}, by the H\"older inequality for traces. Hence smoothness of $\mathcal{E}_{\mathrm{vac}}$ also follows from dominated convergence.
\end{proof}
\section{The semi-classical expansion of the vacuum energy}\label{sect:semicl}
In this section we derive the formula~\eqref{eq:Walecka} as the leading term in an expansion of $\mathcal{E}_{\mathrm{vac}}$ in powers of $\eps$.
The semi-classical expansion is naturally also an expansion in orders of derivatives acting on $\phi$. As such, it reproduces the \enquote{derivative expansions} discussed in the physics literature~\cite{Blunden-90, Chan-86, Cheyette-85, Perry-86}.
For this reason we make very strong smoothness assumptions on the field $\phi$ for its derivation. These can be somewhat relaxed if one is only interested in a specific order of the expansion. 

We will also need to assume that $\abs{\phi}< 1$. The lower bound $\phi>-1$ is certainly necessary due to the logarithm in formula~\eqref{eq:Walecka}, and the upper bound is needed to obtain the following series expansion:
\begin{equation}\label{eq:I Neumann}
  \int \Tr I(\omega,\phi)\,\ud \omega =
\int \Tr\Big( \sum_{j=5}^\infty \frac{(-1)^{j+1}}{j} \big(\beta\phi(D_0+\ui \omega)^{-1}\big)^j\Big)\, \ud \omega\,.
\end{equation}
To prove this formula, we use that $\norm{\beta\phi(D_0+\ui\omega)^{-1}}_{\mathfrak{S}_\infty} < 1$ and write
\begin{align*}
\Tr\left[I(\omega, \phi)\right]
&=\ui\omega \Tr\Big[\big(1+\beta\phi(D_0+\ui \omega)^{-1}\big)^{-1}(D_0+\ui \omega)^{-1}\big(\beta\phi(D_0+\ui \omega)^{-1}\big)^5\Big]\\
&=\ui\omega \Tr\Big[(D_0+\ui \omega)^{-1}\big(\beta\phi(D_0+\ui \omega)^{-1}\big)^5\big(1+\beta\phi(D_0+\ui \omega)^{-1}\big)^{-1}\Big]\\
&=\ui\omega \sum_{j=0}^\infty (-1)^j \Tr\Big[(D_0+\ui \omega)^{-1}\big(\beta\phi(D_0+\ui \omega)^{-1}\big)^{j+5}\Big]\,.
\end{align*}
Then we note that
\begin{equation*}
\omega \frac{\ud}{\ud \omega} \Tr\big(\beta\phi(D_0+\ui \omega)^{-1}\big)^j
=-\ui \omega j \Tr\Big[(D_0+\ui \omega)^{-1}\big(\beta\phi(D_0+\ui \omega)^{-1}\big)^j\Big]\,,
\end{equation*}
 and integrate by parts, which is justified by the KSS-inequality~\eqref{eq:KSS} and gives~\eqref{eq:I Neumann}.
Now $\sum_{j=5}^\infty \frac{(-1)^{j+1}}{j} x^j=\log(1+x)-r(x)$, with a polynomial of order four $r(x)$. As $I(-\omega)=I(\omega)^*$ we then find
\begin{align}\label{eq:log}
 \mathcal{E}_{\mathrm{vac}}(\phi)=-\frac{1}{4\pi} \int &\Tr\Big[\log\big(1+\beta\phi (D_0+\ui \omega)^{-1}\big)+ r\big(\beta\phi (D_0+\ui \omega)^{-1}\big)\Big]\\
 &+\Tr\Big[\log\big(1+(D_0-\ui \omega)^{-1}\beta\phi \big)+r\big( (D_0-\ui \omega)^{-1}\beta\phi\big)\Big]\,\ud \omega\,.\notag
\end{align}
Under sufficiently strong assumptions on $\phi$ we can represent this quantity as an asymptotic series in $\eps$, using techniques of semi-classical pseudo-differential calculus.
We will mainly refer to~\cite{DimSjo-99} for results on this topic. This book treats only scalar-valued symbols, as opposed to the $\C^{4\times 4}$ valued symbols we use, so we will indicate in which sense the results there translate
to our setting, with only minor modifications of the proofs.
\begin{thm}\label{thm:semicl}
Let $\phi\in \mathscr{C}^\infty$ be such that for every multiindex $\alpha\in \N^3$ there exists a constant $c_\alpha$, with $c_0<1$, satisfying
\begin{equation*}
 \left\vert\frac{\partial^{\abs{\alpha}}}{\partial x^\alpha}\phi(x)\right\vert
 \leq c_\alpha (1+\abs{x}^2)^{-(1+\abs{\alpha})/2}\,.
\end{equation*}
Then $\cE_{\mathrm{vac}}(\phi)$ admits an asymptotic expansion at $\eps=0$ of the form
\begin{equation*}
 \eps^3\mathcal{E}_{\mathrm{vac}}(\phi)=\int_{\R^3}\sum_{k=0}^n \eps^k\mathcal{V}_k(\phi)(x)\,\ud x + \mathcal{O}(\eps^{n+1})\,,
\end{equation*}
where $\mathcal{V}_n(\phi)(x)$ depends on the derivatives of order at most $n$ of $\phi$ at $x$ and 
\begin{align*}
 \mathcal{V}_0(\phi)(x)& =-\frac{1}{4 \pi^2}\big(1+\phi(x)\big)^4\log(1+\phi(x)) -P(\phi(x)) \\
 \mathcal{V}_1(\phi)(x)&=0\,.
\end{align*}
\end{thm}
\begin{proof}
In order to keep track of the dependence on $\omega$ we will use a pseudo-differential calculus depending on $\omega$ as an auxiliary parameter. That is, we define the symbol class $\mathcal{S}^m$ to be those $a\in \mathscr{C}^\infty(\R^3\times \R^3\times \R, \C^{4\times 4})$ satisfying
\begin{equation*}
 \Big\Vert\frac{\partial^\alpha}{\partial x^\alpha} \frac{\partial^\beta}{\partial p^\beta} a(x,p,\omega)\Big\Vert_{\C^{4\times 4}}
\leq C_{\alpha, \beta} \left(1+\abs{x}^2\right)^{\tfrac{m-\abs{\alpha}}{2}} \left(1+\abs{p}^2+\omega^2\right)^{\tfrac{m-\abs{\beta}}{2}}\,,
\end{equation*}
for all multiindices $\alpha, \beta \in \N^3$ and appropriate constants $C_{\alpha, \beta}$. An $\eps$-admissible symbol of order $m\in \mathbb{Z}$ is then a function $\eps\mapsto a(\eps)$ in $\mathscr{C}^\infty([0,\eps_0),\mathcal{S}^m)$ with $a_N:=\tfrac{\ud^N}{\ud \eps^N}\big\vert_{\eps=0}a\in \mathcal{S}^{m-N}$, for which the remainder term of the Taylor expansion of order $N$ at $\eps=0$ is of order $\eps^{N+1}$ in the Fréchét topology of $\mathcal{S}^{m-N-1}$.
Let $A=\beta\phi(x) (D_0 +\ui \omega)^{-1}$. Because $(D_0+\ui\omega)^{-1}$ is the quantisation of $(\alpha\cdot p +\beta +\ui \omega)^{-1}$ there is an $\eps$-admissible symbol $a(\eps)$ of order $-1$ such that $A$ equals the $\eps$-Weyl-quantisation $\Opeps{a}$, i.e.
\begin{equation*}
(Af)(x)=\frac{1}{(2\pi \eps)^3}\int_{\R^6} \ue^{\ui\eps^{-1} \langle x-y, p\rangle} a\big((x+y)/2, p, \omega, \eps\big)f(y)\,\ud y\ud p\,,
\end{equation*}
for every Schwartz-function $f$ (as follows from inspection of~\cite[Proposition 7.7]{DimSjo-99}). The leading term is
\begin{equation*}
 a_0=\beta\phi(x)(\alpha\cdot p +\beta +\ui \omega)^{-1}\in\mathcal{S}^{-1}\,.
\end{equation*}
 Because $\norm{\phi}_\infty<c_0<1$, $A$ has norm less than $c_0$. Since $F(z)=\log(1+z)-r(z)$ is a holomorphic function on the open unit disk, $F(A)$ can be computed using Cauchy's formula. Let $\gamma:=\lbrace z:\abs{z}=(1+c_0)/2\rbrace$. Then we have for $z\in \gamma$
 \begin{equation*}
  \norm{(a_0(x,p)-z)^{-1}}_{\C^{4\times 4}}\leq \frac{2}{1-c_0}\,,
 \end{equation*}
so $(a_0(x,p)-z)^{-1}\in \mathcal{S}^0$, with bounds independent of $z$. This implies that $(A-z)^{-1}=\Opeps{b(z)}$, for an admissible symbol $b(z)$ of order zero, with $b_0(z)= (a_0(x,p)-z)^{-1}$ (cf.~\cite[(8.11)]{DimSjo-99}, note that the correct generalisation of a lower bound on a scalar symbol is an upper bound on the inverse). Thus
\begin{equation*}
 F(A)=\frac{1}{2\pi \ui}  \Opeps{\int_\gamma F(z)b(z)\,\ud z}\,,
\end{equation*}
so $F(A)$ itself is the quantisation of an admissible symbol $f$ of order zero. Its expansion is given by
\begin{equation*}
 f_n=\frac{1}{2\pi \ui}  \int_\gamma F(z)b_n(z)\,\ud z\,,
\end{equation*}
in particular $f_0=F(a_0)$ (the formulas for the higher order terms will differ from those for scalar-valued symbols). We now use this to evaluate the trace and the integral over $\omega$, order-by-order in $n$. For $n=0$ we have
\begin{equation*}
 f_0=F(a_0)=a_0^5\sum_{j=0}^\infty \frac{(-1)^{j}}{j+5}a_0^j\in \mathcal{S}^{-5}\,,
\end{equation*}
since the sum is absolutely convergent. This implies that $f_0$ and its derivatives are elements of $L^1(\R^6, \C^{4\times 4})$, so (cf.~\cite[Theorem 9.4]{DimSjo-99}) $\Opeps{f_0}\in \mathfrak{S}_1$ and
\begin{equation*}
 \Tr[\Opeps{f_0}]=\frac{1}{(2\pi \eps)^3} \int_{\R^6} \tr_{\C^4}[f_0(x,p)] \,\ud p\ud x\,.                                                                                                                                                                                                                                                                                                                                                  \end{equation*}
 The leading order in the expansion of $\mathcal{V}$ is thus given by
 \begin{align}
  \mathcal{V}_0(\phi)&=-\frac{1}{4\pi(2\pi)^3}\int_{\R^3}\int_{\R^3}\tr_{\C^4}\big[F(a_0(x,p))]+ F(a_0(x,p)^*)\big]\,\ud p\ud \omega\notag\\
  &=-\frac{1}{4\pi(2\pi)^3} \int\tr_{\C^4}\big[\log(1+a_0)(1+a_0^*)- r(a_0)-r(a_0^*)\big] \,\ud p \ud \omega\,.\label{eq:TrLog1}
 \end{align}
 Now,
 \begin{align*}
 (1+a_0)(1+a_0^*)&=1+a_0+a_0^* + a_0a_0^*\\
%
% &=1+ \Big(\beta \phi(x) (\alpha\cdot p +\beta +\ui \omega) + (\alpha\cdot p +\beta -\ui \omega)\beta \phi(x) + \phi(x)^2\Big)(p^2+1+\omega^2)^{-1}\\
&=1 + \big(2\phi(x)+\phi(x)^2\big)(p^2+1+\omega^2)^{-1}\,,
\end{align*}
so the $\C^4$-trace is trivial and expanding into power series and setting $\xi=(p,\omega)\in\R^4$ we obtain:
\begin{align*}
 \mathcal{V}_0(\phi)=-\frac{1}{8 \pi^4}\int_{\R^4}\sum_{j=3}^\infty \frac{(-1)^{j+1}}{j}(\xi^2+1)^{-j}(2\phi(x)+\phi^2(x))^j - r'(\phi(x),\xi)\,\ud \xi \,.
\end{align*}
The new remainder term $r'$ is a polynomial of fourth order in $\phi(x)$, and it is integrable in $\xi$, since the first term and the total integrand both are. The $\xi$-integral of the first term evaluates to
\begin{equation*}
 \int_{\R^4}(\xi^2+1)^{-j}\,\ud \xi=\frac{\pi^2}{(j-1)(j-2)}\,,
\end{equation*}
and the sum then yields
\begin{align*}
 \sum_{j=3}^\infty &\frac{(-1)^{j+1}}{j(j-1)(j-2)}(2\phi(x)+\phi^2(x))^j\\
 &= (1+2\phi(x)+\phi^2(x))\log(1+2\phi(x)+\phi^2(x))-r''(\phi(x))\\
 &=2 (1+\phi(x))^4\log(1+\phi(x))-r''(\phi(x))\,,
\end{align*}
with another fourth-order polynomial $r''$. With this we obtain
\begin{equation*}
  \mathcal{V}_0(\phi)=-\frac{1}{4 \pi^2} (1+\phi(x))^4\log(1+\phi(x))+\underbrace{ r''(\phi(x))+\Big(\int_{\R^4} r'(\phi(x),\xi)\,\ud \xi\Big)}_{=:-P(\phi(x))}\,.
\end{equation*}
By construction, $P$ is a fourth order polynomial and since the first four derivatives of the integrand in equation~\eqref{eq:TrLog1} w.r.t. $\phi$ vanish, $P(t)$ is the Taylor polynomial of $-\tfrac{1}{4\pi^2}(1+t)^4\log(1+t)$.

We have thus found the claimed leading order term for $\mathcal{E}_{\mathrm{vac}}(\phi)$ and it remains to estimate the error.
First, note that $f-\sum_{n=0}^4 \eps^n f_n \in \mathcal{S}^{-5}$ and that this implies a bound in $L^1(\R, \mathfrak{S}_1)$ on the quantisation, because
\begin{equation*}
 \Opeps{(1+x^2)^{-5/2}(1+p^2+\omega^2)^{-5/2}}\in L^1(\R, \mathfrak{S}_1)\,,
\end{equation*}
and for every $\sigma\in \mathcal{S}^{-5}$
\begin{equation*}
 \Big\Vert\Opeps{(1+x^2)^{-5/2}(1+p^2+\omega^2)^{-5/2}}\Opeps{\sigma}\Big\Vert_\infty \leq C\,,
\end{equation*}
by the Calderon-Vaillancourt Theorem~\cite[Theorem 7.11]{DimSjo-99}.
This proves that the expansion of $f$ gives an expansion of $\cE_{\mathrm{vac}}(\phi)$ of the claimed form, with
\begin{equation*}
 \mathcal{V}_n(\phi)(x):=\int_\R \int_{\R^3}\tr_{\C^4}[f_n(x,p,\omega)]\,\ud p \ud \omega\,,
\end{equation*}
for the terms of order at least five.
It thus remains to estimate $f_n$ for $1\leq n\leq 4$. For brevity, we will only perform the proof for $f_1$, as the arguments for higher order terms are essentially the same. The explicit form of $b_1$, and thereby $f_1$, is obtained from the formula for the Weyl-Moyal product (see %e.g.~\cite[Proposition A.10]{Te03},
~\cite[Theorem 7.3]{DimSjo-99}), which represents the product of operators on the space of admissible symbols. Denoting this product by $\circ$ and by $\lbrace \cdot, \cdot\rbrace$ the Poisson-bracket we have
\begin{align*}
 b_1(z)&=-b_0(z)a_1 b_0(z) - (b_0(z) \circ (a_0-z))_1 b_0(z)\\
&=-(a_0-z)^{-1}a_1(a_0-z)^{-1} - \tfrac{\ui}{2}\lbrace(a_0-z)^{-1}, a_0\rbrace (a_0-z)^{-1}\,.
\end{align*}
Now for every $\sigma\in \mathcal{S}^m$ we have
\begin{equation*}
 [\sigma, (a_0-z)^{-1}]=-(a_0-z)^{-1} [\sigma, a_0] (a_0-z)^{-1} \in \mathcal{S}^{m-1}\,,
\end{equation*}
so the first term of $b_1$, for example, can be written as
\begin{equation*}
b_0(z)a_1b_0(z)=a_1 b_0(z)^{-2} +[a_0, a_1] b_0(z)^{-3} + [a_0, [a_0, a_1]] b_0(z)^{-4}+\mathcal{S}^{-5}\,.
\end{equation*}
On the other hand, we have
\begin{align*}
 \frac{1}{2\pi \ui}  \int_\gamma F(z) \sigma (a_0-z)^{-\ell}\,\ud z
&= \sigma \frac{1}{2\pi \ui}\int_\gamma F(z)\frac{1}{(\ell-1)!}\frac{\ud^{\ell-1}}{\ud z^{\ell-1}} (a_0-z)^{-1}\,\ud z\\
&=\sigma \frac{(-1)^{\ell-1}}{(\ell-1)!}\left(\frac{\ud^{\ell-1}}{\ud z^{\ell-1}}F\right)(a_0)\,,
\end{align*}
 and since $\left(\frac{\ud^{\ell}}{\ud z^{\ell}}F\right)(a_0)\in \mathcal{S}^{-5+l}$ for $\ell\leq 5$ this implies that
\begin{equation*}
 \frac{1}{2\pi \ui}  \int_\gamma F(z)(a_0-z)^{-1}a_1(a_0-z)^{-1}\,\ud z\in \mathcal{S}^{-5}\,.
\end{equation*}
The same argument, applied to the second term of $b_1$, gives $f_1\in  \mathcal{S}^{-5}$, which proves the desired bound.

To complete the proof, we still need to check that $\mathcal{V}_1=0$.
We have
\begin{equation*}
\tr_{\C^4}\Big[\frac{1}{2\pi \ui}  \int_\gamma F(z)(a_0-z)^{-1}a_1(a_0-z)^{-1}\,\ud z\Big]=\tr_{\C^4}[a_1 F'(a_0)]
\end{equation*}
and
\begin{equation*}
 a_1(x,p)=\big(\beta\phi(x) \circ (\alpha p+ \beta +\ui \omega)^{-1}\big)_1=-\tfrac{\ui}{2}\lbrace \beta\phi, (\alpha p+ \beta +\ui \omega)^{-1}\rbrace=
-\tfrac{\ui}{2}\nabla_x \nabla_p a_0\,.
\end{equation*}
Thus, for every $x$ with $\phi(x)\neq 0$,
\begin{equation*}
 \tr_{\C^4}[a_1 F'(a_0)](x,p)=-\tfrac{\ui}{2}\nabla_p \tr_{\C^4}[(\nabla\log\phi(x)) F(a_0(x,p))]\,,
\end{equation*}
the $\ud p$-integral of which vanishes. For the second term of $b_1$, the $\C^4$-trace equals zero, because $\nabla_x a_0$ commutes with $a_0$ and thus
\begin{align*}
\tr_{\C^4}\big((\nabla_x (a_0-z)^{-1}) &(\nabla_p a_0)  (a_0-z)^{-1}\big)=\\
&=-\tr_{\C^4}\big((a_0-z)^{-1}(\nabla_x a_0) (a_0-z)^{-1} (\nabla_p a_0)  (a_0-z)^{-1}\big)\\
&=\tr_{\C^4}\big((a_0-z)^{-1}(\nabla_x a_0) (\nabla_p (a_0-z)^{-1})\big)\\
&=\tr_{\C^4}\big((\nabla_p (a_0-z)^{-1}) (\nabla_x a_0)(a_0-z)^{-1})  \big)\,.
\end{align*}
This concludes the proof of Theorem~\ref{thm:semicl}
\end{proof}
\section{Applications}
\subsection{An interacting Dirac Klein-Gordon model}\label{sect:app}
In this section we discuss the implications of the results of Section~\ref{sect:semicl} for a simple model of a particle, described by a Dirac equation, interacting with a classical field, described by a Klein-Gordon equation. Such models arise in the mean-field approximation of models in nuclear physics. They are often treated with the additional \enquote{no sea} approximation, which amounts to neglecting the contribution of negative energy states.
In such models, there are usually additional vector-fields to consider. However, due to gauge-invariance of the Dirac equation, the leading order of the vacuum energy is solely due to the scalar field. With the \enquote{no sea} approximation, the static model amounts to finding critical points of the energy
\begin{align*}
& \mathcal{E}_0\left(\psi, \sigma\right)\\
&:=\int_{\R^3} \overline\psi(x)\left(-\ui c \alpha \cdot \nabla + \beta ( c^2 + g \sigma(x))\right)\psi(x) + \tfrac12 c^2 \abs{\nabla \sigma(x)}^2 + \tfrac12 M^2c^4 \abs{\sigma(x)}^2\ud x\,,
\end{align*}
 under the constraint that $\psi\in L^2(\R^3, \C^4)$ be a normalised function in the positive spectral subspace of $-\ui c \alpha \cdot \nabla + \beta ( c^2 + g \phi(x))$, where we have introduced the physical constants, $c$ (speed of light), $M$ (mass of the field), and $g$ (particle-field coupling constant). If we rescale lengths by $\eps=\sqrt{c}$, choose large couplings of the form $g=g_0 \eps^{-7/2}$ and set $\phi(x)=\eps^{-3/2}g_0 \sigma(\eps x)$ this becomes
\begin{align*}
\mathcal{E}_0\left(\psi, \phi\right)
=\eps^{-4}\int_{\R^3} \overline\psi(x)&\left(-\ui\eps \alpha \cdot \nabla + \beta ( 1 + \eps^2 \phi(x))\right)\psi(x)\\
&+\tfrac12 g_0^{-2}\eps^2 \left(\eps^2 \abs{\nabla \phi(x)}^2 +  M^2 \abs{\phi(x)}^2\right)\,\ud x\,.
\end{align*}
Writing $\psi=(\varphi, \chi)$ with $\varphi, \chi\in H^1(\R^3, \C^2)$, the Euler-Lagrange equations with a Lagrange multiplier $1-\eps^2 \mu$ are given by
\begin{equation*}
 \left \lbrace
 \begin{aligned}
-\ui \eps \sigma\cdot \nabla \chi &=-\eps^2 \mu \varphi-\eps^2\phi\varphi\\
  -\ui \eps \sigma\cdot \nabla \varphi&=(2-\eps^2\mu+\eps^2\phi\varphi)  \chi\\
  -\eps^2 \Delta \phi + M^2\phi& = -g_0^2 (\varphi^2 - \chi^2)\,.
\end{aligned}
  \right.
\end{equation*}
For $\eps$ small enough, these equations are known~\cite[Theorem 3]{LewRot-14} to have a solution $(\varphi_\eps, \chi_\eps, \phi_\eps)$, for which
$
 \varphi_\eps \to v\psi_{\mathrm{NLS}}
 %\,,\qquad
 %\eps^{-1}\chi_\eps\to
$
in $H^2(\R^3,\C^2)$, where $v\in S^3\subset \C^2$ is an arbitrary vector and $\psi_{\mathrm{NLS}}$ is the unique positive solution to the non-linear Schr\"odinger equation
\begin{equation}\label{eq:NLS}
 -\tfrac12 \Delta\psi -\left(\tfrac{g_0}{M}\right)^2 \abs{\psi}^2 \psi + \mu\psi=0\,.
\end{equation}
% was ist die norm der Lsg? F'llt diese exp. ab?
%For a given value of $g_0/M$, the $L^2$-norm of $\psi_{\mathrm{NLS}}$ depends on $\mu$, and
By a scaling argument, there exists a unique $\mu$ for which $\norm{\psi_{\mathrm{NLS}}}_{L^2}=1$, and the corresponding solution $(\varphi_\eps, \chi_\eps, \phi_\eps)$ is a critical point of $\mathcal{E}_0$ satisfying the constraints.
If we apply the same procedure to the energy $\mathcal{E}$ including the contribution of the \enquote{sea} of negative energy states, defined as in Section~\ref{sect:reg}, we obtain
\begin{equation*}
\mathcal{E} (\psi, \phi)=\mathcal{E}_0 (\psi, \phi)+\eps^{-4}\mathcal{E}_{\mathrm{vac}}(\eps^2 \phi)\,.
\end{equation*}
Remember that we have shown in Proposition~\ref{prop:Evac} that $\mathcal{E}_{\mathrm{vac}}(\phi)$ is smooth in an $H^1$-neighbourhood of $\phi=0$, and that its first four derivatives at this point vanish. Hence the contribution of $ \mathcal{E}_{\mathrm{vac}}(\eps^2 \phi)$ will be small, of order $\eps^3$, for $\eps\to 0$, despite the pre-factor of $\eps^{-4}$. It is thus easy to show, by the same implicit-function argument used in~\cite{LewRot-14}, that this energy also has a critical point $(\psi_\eps, \phi_\eps)$, exhibiting the same limit, as $\eps\to 0$, as without the vacuum contribution.
For the term $\mathcal{E}_{\mathrm{vac}}(\eps^2 \phi)$ both an expansion around $\phi=0$ and the semi-classical expansion of Section~\ref{sect:semicl} are appropriate.
The second is justified by the fact that the solution to the NLS equation~\eqref{eq:NLS} is smooth and exponentially decreasing (see~\cite{Strauss-77}), so the hypothesis of Theorem~\ref{thm:semicl} are fulfilled in this limit.
\subsection{The stability of homogeneous systems}\label{sec:matter}\
In this section we discuss the effects of vacuum polarisation on the stability of homogeneous systems. For uniform nuclear matter, stability was already studied in the early works~\cite{ChiWal-74,Chin-76,Chin-77} and many later contributions, but with an emphasis on the dominant effect of the vector mesons and neglecting vacuum polarisation.

In this section we present formal derivations and numerics, we will indicate how some of these can be made rigorous.
We take the state of the system to be of uniform density and chemical potential $\mu>0$,
\begin{equation*}
 \gamma_\mu=\chi_{(-\infty, \mu]}(D_\phi)\,.
\end{equation*}
The energy of the positive energy states is then formally given by
\begin{equation*}
 \cE_0(\mu, \phi)=\frac{1}{2g^2} \int_{\R^3} \left(M^2 \abs{\phi}^2 + \abs{\nabla \phi}^2\right) \,\ud x + \Tr(\chi_{(0, \mu]}(D_\phi))\,,
\end{equation*}
where $M$ is the mass of the field and $g$ the coupling constant. The total energy is the sum $\cE_0(\mu, \phi) +\cE_{\mathrm{vac}}(\phi)$. 
%To properly define these quantities for $\phi$ that does not decay one should consider a system in a box of finite volume and take the thermodynamic limit.
%, where the vacuum energy $\cE_{\mathrm{vac}}$ is given by~\eqref{eq:def_E_vac} also using the trace per unit volume. 
For a constant field, the energy per unit volume becomes 
\begin{equation}\label{eq:En const}
 \cE(\mu, \phi)=\frac{M^2}{2 g^2} \abs{\phi}^2 + \frac{4}{(2\pi)^3} \int\limits_{p^2\leq \mu-(1+\phi)^2} \sqrt{p^2 + (1+\phi)^2} \,\ud p + \mathcal{V}(\phi)\,,
\end{equation}
with $\mathcal{V}(\phi)$ given by equation~\eqref{eq:Walecka}.

We define $\phi_{\mathrm{vac}}\in (-1,0]$ to be the local minimum of~\eqref{eq:En const} in this interval and $\phi_0\in (-1,0]$ to be that of $\cE-\cV$ (see Figure~\ref{fig:fields}). Note that $\phi_{\mathrm{vac}}$ is only a local minimum, as $\mathcal{V}$ is unbounded from below for positive values of $\phi$. We will now study the response of the field $\phi_{\mathrm{vac}}$ to a local perturbation $\eta$ of this system by calculating the response function $T$, which is formally given by
\begin{equation*}
\frac{\ud^2}{\ud \delta^2}\bigg\vert_{\delta=0} \mathcal{E}(\mu, \phi_{\mathrm{vac}} + \delta \eta)=\frac{1}{(2\pi)^d}\int_{\R^3} \abs{\hat\eta(p)}^2 T(p,\mu)\,\ud p\,,
\end{equation*}
and compare this to the corresponding results for $\phi_0$.
This expression is only formal since $\phi_{\mathrm{vac}} + \delta \eta\notin H^1$ and $\mathcal{E}_{\mathrm{vac}}(\phi_{\mathrm{vac}} + \delta \eta)$ is not defined. In order to give a proper definition, one should consider a system in a finite box and take the thermodynamic limit of the difference $\cE(\mu, \phi_{\mathrm{vac}} + \delta \eta)-\cE(\mu, \phi_{\mathrm{vac}})$, with $\eta\in H^1$.
We will not perform this construction here, the function $T$ will give a finite expression on the right hand side for sufficiently regular $\eta$.

The response of the vacuum can be obtained by derivation of the formula (cf.~\eqref{eq:log})
\begin{equation}\label{eq:log_resp}
 \mathcal{E}_{\mathrm{vac}}(\phi+\delta\eta)=-\frac{1}{2\pi}\int_\R \tr \Big(\log\big(1+\delta\beta\eta(D_0+\beta \phi +\ui \omega)^{-1}\big)-r(\phi+ \delta\eta, \omega)\Big)\,\ud \omega\,.
\end{equation}
Calculations similar to those of~\cite[Lemma 4.2]{GraHaiLewSer-13}, the details of which are given in the Appendix, give the result
\begin{align*}
 T_{\mathrm{vac}}(p)=&(2\pi)^3\mathcal{V}''(\phi_{\mathrm{vac}})+ 2\pi p^2 \int_0^1 u(1-u)\log\left(\frac{1+u(1-u)p^2}{(1+\phi_{\mathrm{vac}})^2 + u(1-u)p^2}\right)\,\ud u\\
 &+2\pi p^2 \int_0^1\begin{aligned}[t]&\frac{\phi_{\mathrm{vac}} u(1-u)}{(1+u(1-u)p^2)^2}
   \Big(2-6u+2up^2-8u^2p^2+6u^3p^2\\
  &+\phi_{\mathrm{vac}}(1+2up^2-7u^2p^2+10u^3p^2-5u^4p^2)\Big)\,\ud u\,.
 \end{aligned}
\end{align*}
The complicated form of this expression is largely due to the greater number of regularising terms as compared to~\cite{GraHaiLewSer-13}.
The response of the positive energy states is given by
\begin{align*}
 &-2 \pi \ui T_+(\mu, p)\\
 &=\int\limits_{c_\mu}\int\limits_{\R^{3}}\tr_{\C^{4}}\big[\beta(\alpha q + \beta(1+\phi_{\mathrm{vac}}) + z)^{-1}\beta(\alpha (q-p) + \beta(1+\phi_{\mathrm{vac}}) + z)^{-1}\big]\,\ud q \ud z\,,
\end{align*}
where $c_\mu$ is a path in the resolvent set of the matrices $\alpha p + \beta(1+\phi_{\mathrm{vac}})$ enclosing the interval $(0,\mu]$.
The total response function is then given by
\begin{equation*}
 T(\mu, p)= \frac{M^2}{2 g^2}  + \frac{p^2}{2g^2} +  T_+(\mu, p) + T_{\mathrm{vac}}(p)\,.
\end{equation*}

Numerical evaluation of these quantities, introducing the Fermi momentum $k_F:=\sqrt{\mu^2-(1+\phi_{\mathrm{vac}})^2}$ and choosing the parameters $g=10$ and $M=\tfrac12$ (this choice is suggested by~\cite[Table 3]{Reinhard-89}), gives the plots shown in Figure~\ref{fig:response}.

\begin{figure}
\begin{minipage}{0.45\textwidth}%
 \includegraphics[width=0.9\linewidth]{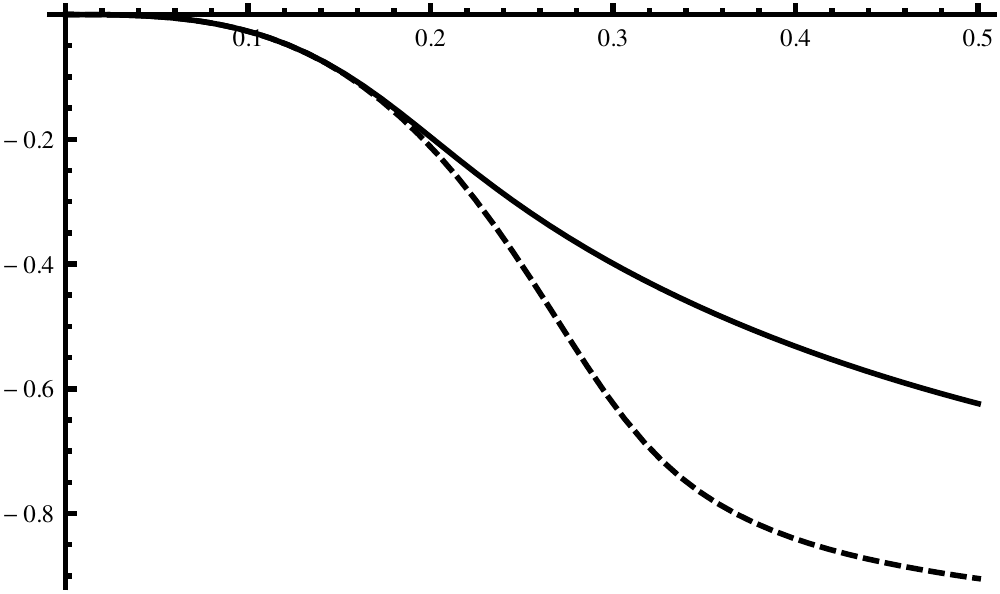}
 \centering
\end{minipage}
\begin{minipage}{0.45\textwidth}%
\centering
 \caption{\label{fig:fields} $\phi_{\mathrm{vac}}$ (solid line) and $\phi_0$ (dashed line) as functions of $k_F=\sqrt{\mu^2-(1+\phi_{0})^2}$ and $k_F=\sqrt{\mu^2-(1+\phi_{\mathrm{vac}})^2}$ (respectively).} 
 \end{minipage}
 \end{figure}
 \begin{figure}
\begin{minipage}{0.45\textwidth}%
\centering
 \includegraphics[width=0.9\linewidth]{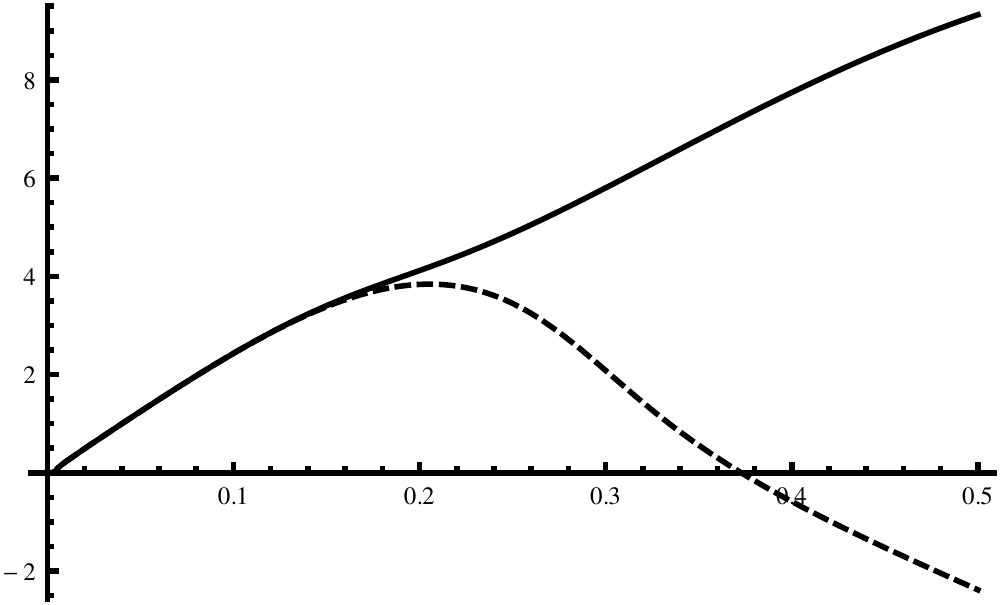}
\end{minipage}
\begin{minipage}{0.45\textwidth}%
 \includegraphics[width=0.9\linewidth]{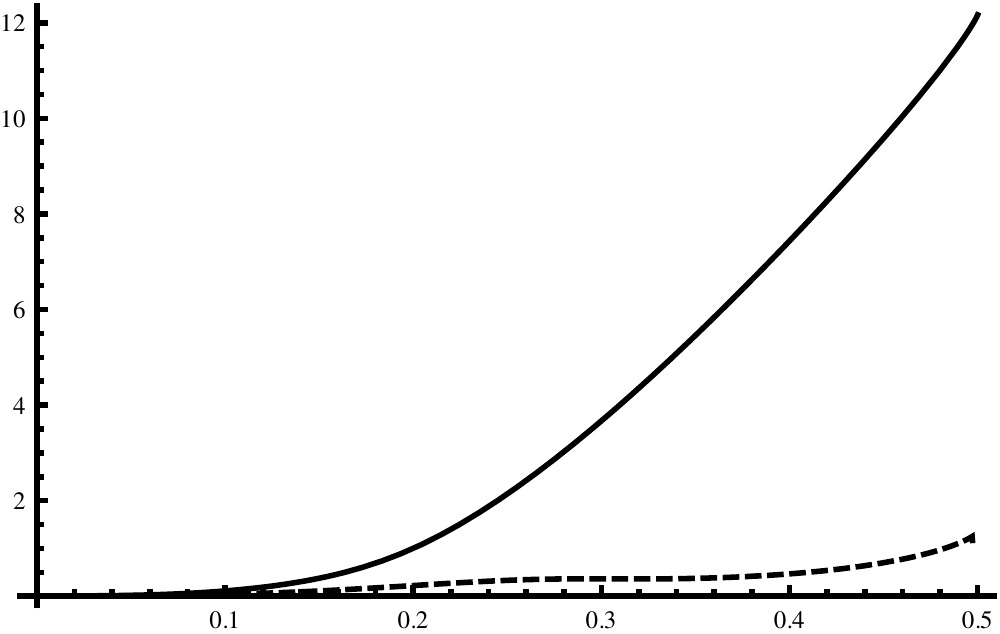} 
 \centering
\end{minipage}
\caption{\label{fig:response} Response functions $T(\mu, p)$ relative to $\phi_0$ (dashed line, without vacuum corrections) and $\phi_{\mathrm{vac}}$ (solid line, including these) as functions of $k_F$, at fixed $p=0.01$ (left) and  $p=1$ (right).}
\end{figure}
We observe that the vacuum polarisation has a twofold stabilising effect: First, it stops $\phi$ from reaching the point $-1$ where the gap in the spectrum of $D_\phi$ closes, and second it renders the minimizer $\phi$ more stable to local perturbations compared to $\phi_0$, which can be unstable at high Fermi momentum. It is an interesting open problem to make these observations rigorous.
\bigskip

\noindent\textbf{Acknowledgment.} We acknowledge financial support from the European Research Council under the European Community's Seventh Framework Programme (FP7/2007-2013 Grant Agreement MNIQS 258023). 
\appendix
\section*{Appendix: Calculation of the vacuum response function}
We start by taking the derivative of the integrand in~\eqref{eq:log_resp}, writing $\phi=\phi_{\mathrm{vac}}$,
\begin{equation*}
\frac{\ud^2}{\ud \delta^2}\bigg\vert_{\delta=0}\tr \Big(\log\big(1+\delta\beta\eta(D_0+\beta \phi +\ui \omega)^{-1}\big)-r(\phi + \delta\eta, \omega)\Big)
\end{equation*}
with
\begin{equation*}
r(\phi+\delta\eta, \omega)=\sum_{j=1}^{d+1}\frac{(-1)^{j+1}}{j}\Big( \beta(\phi + \delta\eta) (D_0 +\ui \omega)^{-1}\Big)^j\,.
\end{equation*}
We have
\begin{align*}
\frac{\ud^2}{\ud \delta^2}\bigg\vert_{\delta=0}\tr&\Big(\log\big(1+\delta\beta\eta(D_0+\beta \phi +\ui \omega)^{-1}\big)\Big)
%
%&=-\tr\Big(\beta \eta (D+\beta\phi +\ui\omega)^{-1}\big(1+\beta\eta(D+\beta\phi+\ui \omega)^{-1}\big)\Big)^2\\
%
=-\tr\Big(\beta\eta(D_0+\beta\phi +\ui \omega)^{-1})\Big)^2
\end{align*}
and
\begin{align*}
&\frac{\ud^2}{\ud \delta^2}\bigg\vert_{\delta=0}\tr \big(r(\phi+ \delta\eta, \omega)\big)\\
&=\begin{aligned}[t]&\tr\Big[-\big(\beta\eta(D+\ui \omega)^{-1}\big)^2+2\phi \beta (D+\ui \omega)^{-1}\big(\beta\eta(D+\ui \omega)^{-1}\big)^2\\
&-\phi^2\Big( \big(\beta(D+\ui \omega)^{-1}\beta\eta(D+\ui \omega)^{-1}\big)^2+2 \big(\beta (D+\ui \omega)^{-1}\big)^2\big(\beta\eta(D+\ui \omega)^{-1}\big)^2 \Big)\Big]\,,
\end{aligned}
\end{align*}
which is only formal as the terms are not individually trace-class, but can be made rigorous by inserting $\pm\tilde r(\delta \eta, \omega)$ which is the remainder for a Dirac operator $\widetilde D_0$ with mass $1+\phi$.

Now let $f,g\in \mathscr{C}(\R^3,\C^4\times \C^4)$ be sufficiently regular and decay at infinity. Then the operator $\eta g(D_0+\ui \omega)\eta f(D_0+\ui \omega)$ has a continuous integral kernel and
\begin{align}
 &\tr\big(\eta g(D_0+\ui \omega)\eta f(D_0+\ui \omega)\big)\notag\\
 %
%  &=\frac{1}{(2\pi)^{2d}}\int \tr_{\C^s(d)}\Big[\eta(x) g(\alpha p+\beta + \ui \omega)\eta(y) f(\alpha q +\beta +\ui \omega)\Big]\ue^{-\ui y(p-q)}\ue^{-\ui x(q-p)} \,\ud x\ud y \ud p\ud q\,.\notag\\
 %
 &=\frac{1}{(2\pi)^{3}}\int \abs{\hat\eta(p)}^2\tr_{\C^4}\Big[g(\alpha q+\beta + \ui \omega) f(\alpha (q-p) +\beta +\ui \omega)\Big]\,\ud p\ud k\,.\label{eq:tr pseudo}
\end{align}
This gives us (writing $f(q)=(\alpha q + \beta +\ui \omega)^{-1})$ and $g(q)=(\alpha p +\beta(1+\phi) +\ui \omega)^{-1}$)
\begin{align*}
T_{\mathrm{vac}}(p)=\frac{1}{2\pi} \int \tr_{\C^4}
\Big[&\beta g(q)\beta g(q-p) - \beta f(q)\beta f(q-p)+2\phi (\beta f(q))^2 \beta f(q-p)\\
&-\phi^2 \Big((\beta f(q))^2 (\beta f(q-p))^2 + 2 (\beta f(q))^3 \beta f(q-p)\Big)\Big]\ud q \ud \omega\,. 
\end{align*}
Using that $g(q)=(\alpha q + \beta(1+\phi) -\ui \omega)/(q^2+(1+\phi)^2+\omega^2)$ the spinor trace of the first term can be evaluated, yielding
\begin{equation}\label{eq:tr 1+phi}
 \tr_{\C^4}\big[\beta g(q)\beta g(q-p)\big]
 =\frac{4(1+\phi)^2 -4 \omega^2+4qp-4q^2}{(q^2+(1+\phi)^2+\omega^2)(q-p)^2+(1+\phi)^2+\omega^2)}\,.
\end{equation}
The first term involving $f$ gives the same result with \enquote{mass} one instead of $1+\phi$, and for the remaining ones we have
\begin{equation*}
 \Tr_{\C^4}\big[\beta f(q)^{1+j}\beta f(q-p)^{1+k}\big]=\frac{\theta_{j,k}(\omega, q,p)}{(q^2+1+\omega^2)^{1+j}((q-p)^2+1+\omega^2)^{1+k}}\,,
\end{equation*}
with the polynomials
\begin{align*}
 \theta_{1,0}=&4+8pq-12q^2-12\omega^2\\
 \theta_{2,0}=&4+12pq-24q^2-4(pq)q^2+4q^4-24\omega^2+4pq\omega^2+8q^2\omega^2 +4\omega^4\\
 \theta_{1,1}=& 4-4p^2+24pq-24q^2+4p^2q^2-8(pq)q^2+4q^4\\
 &-24\omega^2+4p^2\omega^2-8pq \omega^2 + 8q^2\omega^2 + 4\omega^4\,.
\end{align*}
Using the substitution
\begin{equation*}
 \frac{j!k!}{a^{1+j}b^{1+k}}=\int\limits_0^\infty \int\limits_0^\infty \xi^j \zeta^k \ue^{-a\xi-b\zeta}\,\ud\xi \ud\zeta
 =\int\limits_0^1 \int\limits_0^\infty s^{1+j+k} u^j(1-u)^k \ue^{-s(ua+(1-u)b)}\,\ud s\ud u
\end{equation*}
leads us to the Gaussian integrals
\begin{equation*}
 \int_\R \omega^{2l}\ue^{-s\omega^2}\,\ud \omega\,, \qquad \int_\R q^l_i \ue^{-s(q^2_i-2q_i(1-u)p_i)}\,\ud q_i\,,
\end{equation*}
for $\omega$ and the three components of $q$. Evaluating these gives the terms
\begin{equation*}
 j!k!\pi^2 s^{j+k-1}u^j(1-u)^k \Theta_{j,k}(s,p,u)\ue^{-s(m+u(1-u)p^2)}\,,
\end{equation*}
with $m=1$, or $m=1+\phi$ for the evaluation of~\eqref{eq:tr 1+phi}, and
\begin{align*}
 \Theta_{0,0}=&-\tfrac8s+ 4m^2+4u(1-u)p^2\\
 \Theta_{1.0}=&-\tfrac{24}{s}+4-4p^2+16p^2u-12p^2u^2\\
\Theta_{2,0}=&\tfrac{1}{s}\left(\tfrac{24}{s}-48+8 p^2-24 p^2 u+24 p^2 u^2\right)\\
&+4-4 p^2+24 p^2 u-24 p^2 u^2+4 p^4 u^2-8 p^4 u^3+4 p^4 u^4\\
\Theta_{1,1}=&\tfrac{1}{s}\left(\tfrac{24}{s}+12 p^2-36 p^2 u+24 p^2 u^2-48\right)\\
&-12 p^2+36 p^2 u-4 p^4 u-24 p^2 u^2+12 p^4 u^2-12 p^4 u^3+4 p^4 u^4\,.
\end{align*}
After gathering all of these terms (with the correct pre-factors) we split $T_{\mathrm{vac}}(p)=T_{\mathrm{vac}}(0) + T_{\mathrm{sing}}(p) + T_{\mathrm{reg}}(p)$, where $T_{\mathrm{sing}}$ contains all those summands whose integrand vanishes at $p=0$ and diverges at $s=0$. We then have
\begin{equation*}
 T_{\mathrm{vac}}(0)=(2\pi)^3\mathcal{V}^{''}(\phi)\,,
\end{equation*}
since $p=0$ corresponds to a perturbation by a constant field. Calculating the $\ud s$ integrals for the other terms yields
\begin{align*}
T_{\mathrm{sing}}(p)&=2\pi\int_0^1 \int_0^\infty  u(1-u)p^2 \ue^{-su(1-u)p^2}\frac{\ue^{-s(1+\phi)^2}-\ue^{-s}}{s} \,\ud s\ud u\\
&=2\pi p^2 \int_0^1 u(1-u)\log\left(\frac{1+u(1-u)p^2}{(1+\phi)^2 + u(1-u)p^2}\right) \,\ud u\,,
\end{align*}
and
\begin{align*}
 T_{\mathrm{reg}}(p)
 =2\pi p^2 \int_0^1\begin{aligned}[t]&\frac{\phi u(1-u)}{(1+u(1-u)p^2)^2}
   \Big(2-6u+2up^2-8u^2p^2+6u^3p^2\\
  &+\phi(1+2up^2-7u^2p^2+10u^3p^2-5u^4p^2)\Big)\,\ud u\,.
  \end{aligned}
\end{align*}
% 
% 
% 
% 
% \bibliographystyle{abbrv}
% \bibliography{biblio}

\end{document}